\documentclass{article}

\usepackage{amsmath, amssymb, amsthm}  
\usepackage{graphicx}  
\usepackage{hyperref}  
\usepackage{geometry}  
\usepackage{enumitem}  
\usepackage{booktabs}  
\usepackage{xcolor}  

\usepackage{bm}
\usepackage{bbm}
\usepackage{commath}
\usepackage{hyperref}
\usepackage{cleveref}
\usepackage{nicefrac}
\usepackage{physics}
\usepackage{subcaption}
\usepackage{mathtools}
\usepackage{algpseudocode}
\usepackage{algorithm}
\usepackage{tikz}
\usepackage{tikz-3dplot}
\usepackage{thmtools, thm-restate}

\newcommand{\cX}{\mathcal{X}}

\newcommand{\bbR}{\mathbb{R}}

\newcommand{\Ex}{\mathbb{E}}
\newcommand{\eps}{\varepsilon}

\newcommand{\IGNORE}[1]{}

\newcommand{\R}{\mathbb{R}}
\newcommand{\E}{\mathbb{E}}
\newcommand{\A}{\mathcal{A}}
\newcommand{\X}{\mathcal{X}}

\newcommand{\inner}[2]{\langle #1, #2 \rangle}
\newcommand{\supdiff}{\partial^+}
\newcommand{\subdiff}{\partial^-}
\newcommand{\wbar}{\overline{w}}

\newcommand{\ba}{\bm{a}}

\renewcommand{\ip}[2]{\langle #1, #2\rangle}

\newcommand{\conv}{\mathtt{conv}}

\newtheorem{observation}{Observation}

\newtheorem{theorem}{Theorem}
\newtheorem{lemma}{Lemma}
\newtheorem{remark}{Remark}
\newtheorem{assumption}{Assumption}

\newtheorem{definition}{Definition}
\newtheorem{cor}{Corollary}
\newtheorem*{theorem*}{Theorem}

\newtheorem{proposition}{Proposition}

\newcommand{\bR}{\mathbb{R}}

\let\tilde\widetilde

\newcommand{\LIN}{\mathsf{LIN}}
\newcommand{\FB}{\mathsf{FB}}

\newcommand{\cA}{\mathcal{A}}

\usepackage[suppress]{color-edits}
\addauthor[sz]{sz}{red}
\usepackage{natbib}

\title{Distributional Robustness of Linear Contracts}
\author{Shiliang Zuo \\ szuo.rs@gmail.com}
\date{}

\begin{document}

\maketitle

\begin{abstract}
Linear contracts—payments proportional to observable outcomes—are ubiquitous in practice, yet optimal contract theory often prescribes complex, nonlinear structures. We provide a distributional robustness justification for linear contracts. We study a principal-agent problem where the agent exerts costly effort across multiple tasks, generating a stochastic signal upon which the principal conditions payment. The principal faces distributional ambiguity: she knows the expected signal for each effort level, but not the full distribution. She seeks a contract maximizing her worst-case payoff over all distributions consistent with this partial knowledge.
Our main result shows that linear contracts are optimal for such a principal. For any contract, there exists a linear contract achieving weakly higher worst-case payoff. The proof introduces the concavification approach built around the notion of self-inducing actions—actions where an affine contract simultaneously induces the action as optimal and supports the concave envelope of payments from above. We show that self-inducing actions always exist as maximizers of the gap between the concave envelope and agent's cost function. 
We extend these results to multi-party settings. In common agency with multiple principals, we show that affine contracts improve all principals' worst-case payoffs. In team production with multiple agents, we establish a complementary necessity result: if any agent's contract is non-affine, the unique ex-post robust equilibrium is zero effort. Finally, we show that homogeneous utility and cost functions yield tractable characterizations, enabling closed-form approximation ratios and a sharp boundary between computational tractability results. 
\end{abstract}

\section{Introduction}
\label{sec:intro}

Contract theory studies how a principal can motivate a self-interested agent to exert effort through outcome-contingent incentives. Such contracts are pervasive across economic settings---from sales commissions and executive compensation to sharecropping and revenue-sharing agreements. In principle, optimal contracts can be highly complex, specifying a distinct payment for each possible observable outcome. In practice, however, linear contracts---payments proportional to observable outcomes---are remarkably prevalent. What justifies this simplicity? One answer lies in \emph{robust contract design}: \cite{carroll2015robustness} showed that linear contracts are max-min optimal when the principal faces ambiguity about the agent's technology set---which actions are available to the agent.

In this work we study linear contracts from a distributional robustness lens. We study a principal-agent problem where the agent exerts costly effort $a$ that generates a stochastic signal $x$, upon which the principal conditions payment. The principal designs a contract $w(x)$ to maximize her payoff—the expected benefit from the agent's effort minus the expected payment. The principal faces \emph{distributional ambiguity}: she knows the expected signal for each effort level, $\Ex[x \mid a]$, but not the full distribution of signals. She seeks a contract that maximizes her worst-case payoff over all distributions consistent with this partial knowledge.

Our ambiguity concerns a different primitive compared with \cite{carroll2015robustness}: not what the agent can do, but how effort maps to stochastic outcomes. Our setting is motivated by \emph{distributional robustness} in optimization and mechanism design: when a decision-maker knows only partial distributional information (moments, support, or shape constraints), she optimizes against worst-case distributions consistent with this knowledge \citep{bergemann2005robust, carrasco2018optimal, bachrach2022distributional}. Our setting imports this paradigm into contracting with moral hazard. 

\subsection{Our Setting and Results} 
We study a principal who contracts with an agent exerting effort across $d$ tasks. The agent's effort $a$ generates a stochastic $d$-dimensional signal $x$, upon which the principal conditions payment via a contract $w(x)$. The principal faces \emph{distributional ambiguity}: she knows the expected signal for each effort level, $\Ex[x \mid a] = a$, but not the full distribution. She seeks a contract maximizing her worst-case payoff over all distributions consistent with this partial knowledge.


Our main result (\Cref{thm:robustness}) shows that linear contracts are optimal for such a principal: for any contract $w$, there exists a linear contract $\phi$ with weakly higher worst-case payoff: $V_P(\phi) \geq V_P(w)$.


\paragraph{The concavification approach.} The starting observation is that affine contracts are \emph{distribution-invariant}: the expected payment $\Ex[\alpha + \beta^\top x \mid a] = \alpha + \beta^\top \Ex[x \mid a]$ depends only on the mean. This immediately suggests that affine contracts are natural candidates for robust optimality—their performance is unaffected by distributional ambiguity. The question is whether, for any contract $w$, one can find an affine contract that dominates it in worst case.

In one dimension, the answer is immediate. Suppose the outcome space were $[0, x_{\max}]$, then any distribution with mean taking value in $[0, x_{\max}]$ can be achieved by a two-point distribution supported on the boundary. So for any contract $w$, the affine contract interpolating $w$ at $0$ and $x_{\max}$ matches $w$'s expected payment under the boundary distribution, and by distribution-invariance, performs at least as well under all other compatible distributions. One of the main result of \cite{dutting2019simple} was in fact the robustness of linear contracts when outcomes are scalars, for which they gave a more complicated proof. However their proof still relied on the ordering of the real line and existence of the extreme points, neither of which extends to higher dimensions. Hence, it is unclear to the author as to whether their proof is necessary or generalizable. 

In multiple dimensions, the set of distributions compatible with a given mean is vastly richer, and there is no canonical set of extreme points. We develop an approach based on \emph{concavification} to handle this complexity. The key object is the \emph{concave envelope} $\bar{w}$ of the contract $w$: for any action $a$, $\bar{w}(a)$ equals the maximum expected payment the agent can extract when choosing $a$, over all compatible distributions. The central concept is a \emph{self-inducing action}—an action $a^\sharp$ where some affine contract $\psi$ simultaneously (i) induces $a^\sharp$ as the agent's optimum, and (ii) supports $\bar{w}$ from above at $a^\sharp$. When these conditions align, $\psi$ replicates $w$'s performance under a distribution placing mass where $\psi$ touches $w$, and distribution-invariance implies $\psi$ dominates $w$ in worst case. The crux is that self-inducing actions \emph{always exist}: they are precisely the maximizers of $\bar{w}(a) - c(a)$, where $c$ is the agent's cost function.

\paragraph{Extensions to multi-party settings.} The concavification approach extends beyond bilateral contracting. In \emph{common agency}, where multiple principals contract with a shared agent, we apply concavification to the aggregate contract. The self-inducing action for the aggregate decomposes to yield affine contracts for each principal, and we show that for any contract profile $(w_{P_1}, w_{P_2}, \dots)$ used jointly by all principals, there exists an affine profile improving all principals' worst-case payoffs (\Cref{thm:robust-common-agency}). In \emph{team production}, where a principal contracts with multiple agents, we obtain a complementary \emph{necessity} result: when agents themselves face distributional uncertainty, non-affine contracts create payment uncertainty that drives ambiguity-averse agents toward zero effort. If any agent's contract is non-affine, the unique ex-post robust equilibrium is zero effort (\Cref{thm:team-robust}).

\paragraph{Tractability under homogeneity.} While optimal linear contracts are generally difficult to characterize, we show that homogeneous utility and cost functions yield tractable characterizations; based on these, we derive approximation guarantees and computational results for linear contracts across all three settings.

\subsection{Related Work}

\paragraph{Contract design and Robustness.}
The study of contract design is a foundational topic in microeconomics, dating back to \citet{holmstrom1979moral}. Other work has examined related and richer settings such as multitasking \cite{holmstrom1991multitask}, common agency \cite{bernheim1986common}, and team production \cite{holmstrom1982moral}, which provide key background for the setting we study here.

Our robustness results builds on the foundational study of robust contracts by \citet{carroll2015robustness}, who considers a principal facing ambiguity about the agent's technology set---the set of action-outcome distributions available to the agent. Carroll shows that linear contracts are max-min optimal when the principal knows only a subset of available technologies. This robustness rationale for linear contracts has been extended to richer environments: \citet{marku2024robust} studies common agency with technology set ambiguity, and \citet{dai2022robust} extends the framework to team production settings. Our work complements this literature by considering a different source of ambiguity: rather than uncertainty about which technologies the agent can access, we study uncertainty about the distributional mapping from effort to outcomes, assuming the principal knows only the mean outcome structure. This framework was originally introduced by \citet{dutting2019simple} in the single-dimensional setting; our work generalize this to the multitasking scenario as well as common agency and team production settings. This distributional robustness perspective yields similar conclusions---optimality of linear contracts---but via different techniques based on concavification. 


\paragraph{Distributionally robust mechanism design.}
Our distributional robustness approach connects to the broader literature on robust mechanism design (e.g., \citep{bergemann2005robust}). Particularly relevant is work on distributionally robust mechanism design, where the designer knows only partial information about bidders' value distributions \citep{bachrach2022distributional, carrasco2018optimal}. These papers typically assume the designer knows moments or support of the distribution and optimizes against worst-case realizations. Our setting differs in that we study bilateral contracting with moral hazard rather than auction design with adverse selection, but the max-min methodology is similar. 


\paragraph{Computational aspects of contract design.}
A growing literature studies computational questions in contract design with combinatorial structures. For single-agent settings, \citet{dutting2024combinatorial}, \cite{dutting2025combinatorial} and \cite{dutting2024combinatorial} analyze the complexity of computing optimal contracts with combinatorial action spaces. Multi-agent contract design introduces additional computational challenges, this problem was explored in \cite{dutting2023multi}, \cite{dutting2025multi}, \cite{deo2024supermodular}. These papers focus on discrete action spaces. Our computational results on team production complements this line of work in the sense that we study continuous action spaces, and offers a partial characterization on the tractability boundaries. Besides computational aspect of contract design, several works (\cite{dutting2019simple}, \cite{castiglioni2021bayesian}, \cite{alon2023bayesian}) also study approximation guarantees of linear contracts in the bilateral single-task setting. 

\subsection{Paper organization.} \Cref{sec:problem-statement} introduces the bilateral multitasking model. \Cref{sec:robustness} develops concavification and proves our main result. \Cref{sec:multilateral} extends to common agency and team production. \Cref{sec:homogeneous} characterizes solutions under homogeneity. Proofs appear in the appendix.

\section{Bilateral Principal-Agent Problem}
\label{sec:problem-statement}
In this section we describe a basic setup of a bilateral principal-agent problem with a single principal and single agent. The two parties enters a contractual relation. The agent's effort (or action) space is $\mathcal{A}$. We assume $\mathcal{A}$ to be compact and downward-closed. When the agent takes action $a \in \mathcal{A}$, he incurs a private cost $c(a)$. We assume the cost function $c: \mathcal{A} \to \mathbb{R}^+$ is strictly increasing, strictly convex, and twice differentiable on its domain, with $c(0) = 0$. We further assume the agent's action space $\cA$ to be downward closed. The agent's effort generates a stochastic signal $x \in \mathcal{X}$, observable by both parties. We denote $F(a)$ as the mapping from the action space $\cA$ to a distribution over the signal space $\mathcal{X}$. The signal space can either be discrete or continuous, our result works either way. In the multitask principal-agent setting we have $\mathcal{A}, \mathcal{X} \subset (\mathbb{R}^+)^d$, here $d$ represents the number of tasks, and for an effort vector $a\in\mathcal{A}$, the $i$-component $a_i$ is interpretated as the effort in the $i$-th task, and $x_i$ is interpretated as the signal for the $i$-th task. 

The principal offers a contract $w: \mathcal{X} \to \mathbb{R}^+$, specifying a non-negative payment for each realized signal. The non-negativity constraint reflects limited liability protection for the agent. Given a contract $w(\cdot)$ and the distribution mapping $F(\cdot)$ from actions to the signal space, the agent's expected utility is
\[
V_A(a; w(\cdot), F(\cdot)) = \Ex_{x \sim F(a)}[w(x)] - c(a),
\]
The agent takes an action maximizing his surplus:
\[
a^*(w(\cdot), F(\cdot)) \in \arg\max_a V_A(a; w(\cdot), F(\cdot)). 
\]
We assume ties are broken in favor for the principal. We may simply write $a^*$ when $w$ and $F$ are clear from context. 

Denote $u: \mathcal{A} \to \mathbb{R}^+$ as the principal's gross benefit from action $a$. The principal's expected utility (or surplus) is the expected gross benefit minus the expected payment: 
\begin{align*}
V_P(w(\cdot), F(\cdot)) = u(a^*) - \Ex_{x \sim F(a^*)}[w(x)],
\end{align*}

The timing of the interaction is as follows:
\begin{enumerate}
    \item The principal offers a contract $w: \mathcal{X} \to \mathbb{R}^+$.
    \item The agent chooses effort $a^* \in \mathcal{A}$ to maximize $V_A(a; w(\cdot), F(\cdot))$. 
    \item An outcome $x \in \mathcal{X}$ is realized according to the distribution $F(a)$.
    \item The agent receives payment $w(x)$; the principal receives surplus $u(a) - w(x)$. 
\end{enumerate}

\paragraph{Linear and Affine Contracts}

A \emph{linear contract} is parameterized by $\phi \in (\mathbb{R}^+)^d$ and specifies payment
\[
\phi(x) = \langle \phi, x \rangle = \sum_{i=1}^d \phi_i x_i.
\]

An \emph{affine contract} is parameterized by $\psi = (\psi_0, \psi_1, \ldots, \psi_d) \in \mathbb{R}^{d+1}$ and specifies payment
\[
\psi(x) = \psi_0 + \sum_{i=1}^d \psi_i x_i.
\]

\section{Robustness of Linear Contracts}
\label{sec:robustness}
We consider a principal who faces distributional uncertainty: she knows only the mean outcome $\E[x|a]$ for each action $a$, but not the full distribution of signals. The principal is ambiguity-averse and seeks to maximize her worst-case payoff over all distributions consistent with the known means. 
\begin{definition}[Compatible Distributions]
A distribution mapping $F(a) : \A \mapsto \Delta(\X)$ is \emph{compatible} if for each $a \in \A$,
\[
\E_{x \sim F(a)}[x \mid a] = a.
\]
\end{definition}
\begin{remark}[Generality of the Mean Structure]
Our definition above assumes $\E[x \mid a] = a$, which may appear to require unbiasedness. This is not essential. Our results extend to any known mean structure $\E[x \mid a] = \mu(a)$ satisfying mild regularity conditions---for example, $\mu$ concave in $a$. The proofs require only straightforward modifications. What is crucial is that the principal knows the mean outcome as a function of effort, not necessarily that effort and expected outcome coincide. We only choose this specific form to simplify the exposition.
\end{remark}
The principal's worst-case payoff from contract $w$ is
\begin{equation}\label{eq:worst-case}
V_P(w(\cdot)) = \min_{F \text{ compatible}} V_P(w(\cdot); F(\cdot)).
\end{equation}
The principal's problem is to choose a contract $w$ to maximize this worst-case payoff:
\[
\max_{w: \X \to \R_+} V_P(w(\cdot)).
\]
We can now state our main result.
\begin{theorem}[Robustness of Linear Contracts]\label{thm:robustness}
Suppose $\A$ is compact. For any contract $w : \X \to \R_+$, there exists a linear contract $\phi \in (\R_+)^d$ such that
\[
V_P(\phi) \geq V_P(w).
\]
Further, there exists a linear contract $\phi^*$ maximizing $V_P(\phi)$.
\end{theorem}
The remainder of this section develops the proof. 

    
\paragraph{Step 1: Distribution Invariance of Affine Contracts}
We begin with the following observation.
\begin{observation}\label{obs:invariance}
Let $F$ be a compatible distribution mapping. Under an affine contract $\psi(x) = \psi_0 + \sum_{i=1}^{d} \psi_i x_i$, the agent's expected utility for any action $a$ is
\[
V_A(a; \psi) = \sum_{i=1}^{d} \psi_i \cdot a_i + \psi_0 - c(a),
\]
which depends only on the action $a$ and not on the particular compatible distribution mapping $F$.
\end{observation}
This observation is immediate from the definition of compatible distributions: since $\E_{x \sim F(a)}[x] = a$ for any compatible $F$, the expected payment $\E[\psi(x)] = \psi_0 + \inner{\psi_{1:d}}{a}$ depends only on the mean.
The implication is powerful: the principal's payoff under an affine contract is invariant across all compatible distribution mappings. Therefore, our strategy is to show that for any contract $w(\cdot)$, we can construct an affine contract $\psi$ and a specific compatible distribution $F$ such that $V_P(\psi; F) \geq V_P(w; F)$. By \Cref{obs:invariance}, this implies $V_P(\psi) \geq V_P(w)$, establishing that affine contracts are worst-case optimal.
\paragraph{Step 2: The Concave Envelope}
Let $w(\cdot)$ be any contract. We define the \emph{concave envelope} of $w$ as
\[
\wbar(a) = \sup \left\{ \sum_i \lambda_i w(x_i) : \lambda_i \geq 0, \sum_i \lambda_i = 1, \sum_i \lambda_i x_i = a \right\}.
\]
The concave envelope $\wbar(a)$ represents the maximum expected payment the agent can receive when taking action $a$, over all compatible distribution mappings. This is because any compatible distribution for action $a$ must have mean $a$, and can be expressed as a convex combination of outcomes in $\X$.
For each action $a$, consider the set of supergradients (i.e., the superdifferential) of $\wbar(\cdot)$ at $a$. Each supergradient corresponds to an upper-supporting hyperplane to the graph of $\wbar(\cdot)$. Each hyperplane $H$ that supports $\wbar(\cdot)$ from above naturally defines an affine contract $\psi_H$ satisfying
\[
\psi_H(x) \geq w(x) \quad \text{for all } x \in \X.
\]
Thus, $\psi_H(\cdot)$ upper bounds the original payment function $w(\cdot)$.
\paragraph{Step 3: Defining the Self-Inducing Action}
The key concept in our proof is a \emph{self-inducing action}---an action where the geometry of the concave envelope and the agent's cost function align perfectly.
\begin{definition}[Self-Inducing Action]\label{def:self-inducing}
An action $a \in \A$ is \emph{self-inducing} if there exists an affine contract $\psi$ such that:
\begin{enumerate}[label=(\roman*)]
    \item $\psi$ induces action $a$ (i.e., $a$ is the agent's best response under $\psi$), and
    \item $\psi$ supports $\wbar(\cdot)$ at the point $a$ from above, so that $\psi(x) \geq w(x)$ for all $x$ and $\psi(a) = \wbar(a)$.
\end{enumerate}
\end{definition}

\begin{figure}
    \centering
    \includegraphics[width=0.6\linewidth]{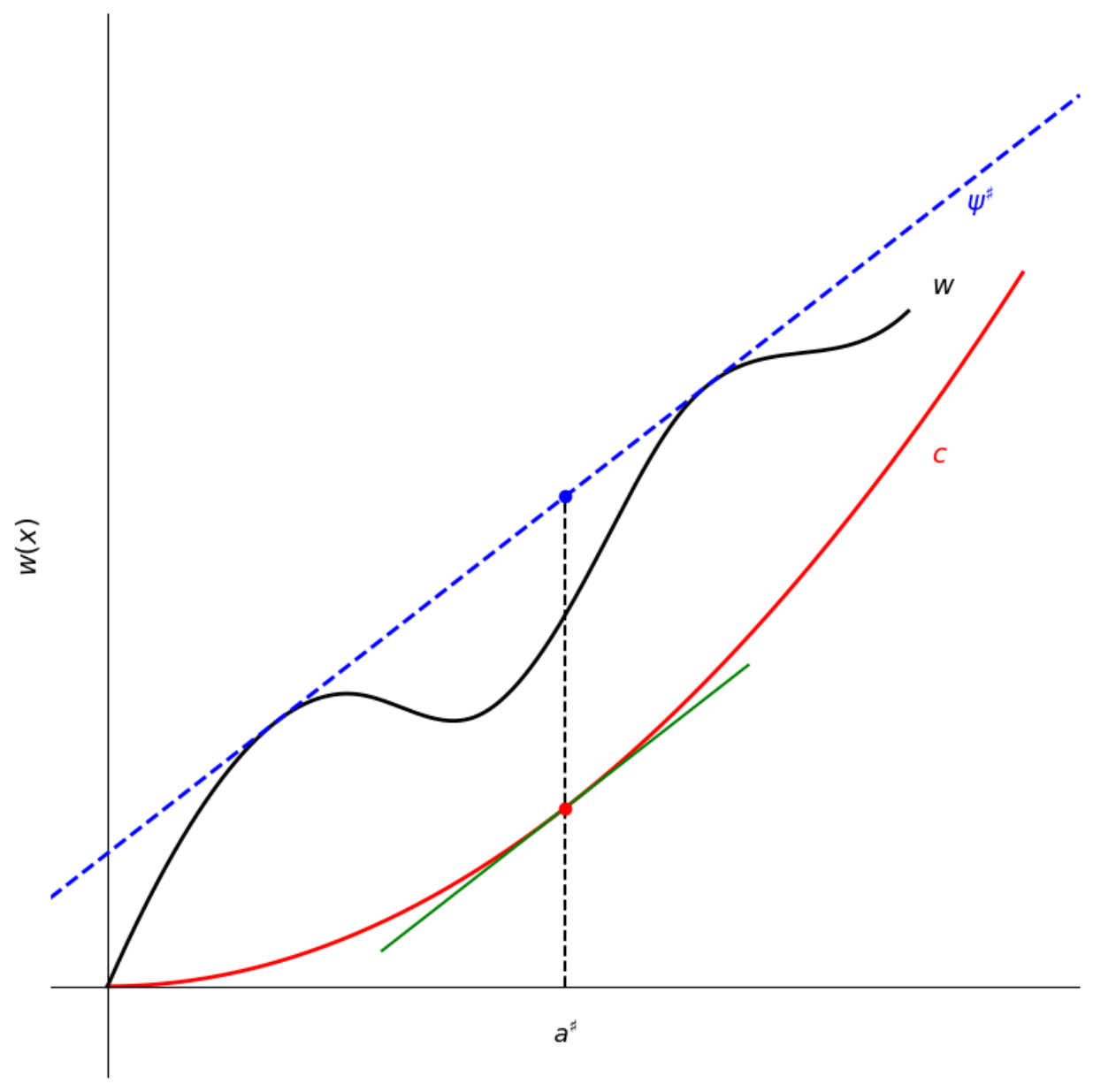}
    \caption{Illustration of self-inducing action. The affine contract $\psi^\sharp$ induces action $a^\sharp$, and bounds the concave envelope of $w$ from above at $a^\sharp$. }
    \label{fig:self-inducing-action}
\end{figure}

\Cref{fig:self-inducing-action} illustrates this concept geometrically: at a self-inducing action $a^\sharp$, the upper-supporting hyperplane $\psi^\sharp$ to the concave envelope $\wbar$ has slope equal to the marginal cost $\nabla c(a^\sharp)$. This alignment ensures that $\psi^\sharp$ both induces $a^\sharp$ (via the agent's first-order condition) and touches $\wbar$ from above at $a^\sharp$.
The following characterization makes the definition concrete.
\begin{proposition}[Characterization of Self-Inducing Actions]\label{prop:characterization}
An action $a \in \A$ is self-inducing if and only if
\[
\subdiff c(a) \cap \supdiff \wbar(a) \neq \varnothing,
\]
where $\subdiff$ and $\supdiff$ denote the subdifferential and superdifferential respectively.
\end{proposition}
\begin{proof}
An affine contract $\psi$ induces action $a$ if and only if $\psi_{1:d} \in \partial c(a)$ (the agent's first-order optimality condition). Meanwhile, $\psi$ corresponds to a supergradient of $\wbar$ at $a$ if and only if $\psi_{1:d} \in \partial \wbar(a)$. The result follows immediately.
\end{proof}
Our first main lemma establishes that self-inducing actions always exist.
\begin{lemma}[Existence of Self-Inducing Action]\label{lem:existence}
For any contract $w : \X \to \R_+$, there exists a self-inducing action $a^\sharp \in \A$.
\end{lemma}

By \Cref{prop:characterization}, we seek an action with $\subdiff c(a) \cap \supdiff \wbar(a) \neq \varnothing$. Consider the convex minimization problem
\[
\min_{a \in \A} \left[ c(a) - \wbar(a) \right].
\]
Since $c$ is strictly convex and $\wbar$ is concave, the objective is strictly convex. Compactness of $\A$ guarantees a minimizer $a^\sharp$. At this minimizer, the first-order condition $0 \in \subdiff[c(a^\sharp) - \wbar(a^\sharp)]$ implies $\subdiff c(a^\sharp) \cap \supdiff \wbar(a^\sharp) \neq \varnothing$, so $a^\sharp$ is self-inducing.
\paragraph{Step 4: Constructing the Dominating Distribution}
Given a self-inducing action, we can construct a specific compatible distribution under which the affine contract performs at least as well as the original contract.
\begin{lemma}[Distribution for Self-Inducing Actions]\label{lem:distribution}
Let $a^\sharp$ be a self-inducing action with corresponding affine contract $\psi^\sharp$ (i.e., $\psi^\sharp$ induces $a^\sharp$ and $\psi^\sharp$ is a supergradient of $\wbar$ at $a^\sharp$). There exists a compatible distribution $F$ such that:
\begin{enumerate}[label=(\roman*)]
    \item The contract $w$ induces action $a^\sharp$ under $F$.
    \item $\E_{x \sim F(a^\sharp)}[w(x)] = \E_{x \sim F(a^\sharp)}[\psi^\sharp(x)]$.
\end{enumerate}
\end{lemma}
The intuition is as follows. Let $\psi^\sharp$ be the affine contract that induces $a^\sharp$ and supports $\wbar(\cdot)$ at $a^\sharp$ from above. At the concave envelope,
\[
\wbar(a^\sharp) = \max \left\{ \sum_k \lambda_k w(x_k) : \lambda_k > 0, \sum_k \lambda_k = 1, \sum_k \lambda_k x_k = a^\sharp \right\}.
\]
Take $a^\sharp = \sum_k \lambda_k x_k$ as the optimal representation realizing $\wbar(a^\sharp)$. The key insight is that $\psi^\sharp(x_k) = w(x_k)$ on all the ``support'' points $x_k$---these are precisely the points where the supporting hyperplane touches the original contract.
We construct the compatible distribution $F(\cdot)$ by placing the distribution for action $a^\sharp$ onto these support points $x_k$. Under this distribution:
\begin{itemize}
    \item The agent receives identical payments under $\psi^\sharp$ and $w$ when taking action $a^\sharp$.
    \item For any other action, the payment under $w$ cannot exceed that under $\psi^\sharp$ (since $\psi^\sharp \geq w$ everywhere).
    \item Since $\psi^\sharp$ induces $a^\sharp$, the contract $w$ also induces $a^\sharp$ under $F(\cdot)$.
\end{itemize}
Therefore, under this distribution, the affine contract $\psi^\sharp$ performs as well as $w$. By \Cref{obs:invariance}, $\psi^\sharp$ achieves the same payoff under all compatible distributions, so $\psi^\sharp$ is no worse than $w$ even under the worst-case distribution.
\paragraph{Step 5: From Affine to Linear Contracts}
The final step shows that any affine contract can be improved to a linear contract with weakly lower payment while inducing the same action. For any action $a$, note that the linear contract $\nabla c(a)$ indeed induces $a$, by the agent's first-order condition when maximizing $\ip{\phi}{a} - c(a)$. One can further show that among all affine contracts, this linear contract induces $a$ at least expected payment.

\section{Beyond the Bilateral Setting}
\label{sec:multilateral}

We extend the concavification approach to two canonical multi-party settings: common agency (multiple principals, one agent) and team production (one principal, multiple agents). 

\subsection{Multiple Principals: Common Agency}
\label{sec:common-agency}

In the common agency setting, $n$ principals contract with a single agent. Each principal $i \in [n]$ offers a contract $w_{P_i}: \mathcal{X} \to \mathbb{R}$ and obtains gross benefit $u_{P_i}(a)$ when the agent exerts effort $a \in \mathcal{A}$, which generates stochastic outcome $x \in \mathcal{X}$. The agent incurs cost $c(a)$ and is protected by \emph{aggregate} limited liability: $\sum_{i=1}^n w_{P_i}(x) \geq 0$ for all $x$. This formulation, following \cite{marku2024robust}, allows cross-subsidization across principals—the agent may pay one principal using payments received from others.

Given contracts $(w_{P_1}, \ldots, w_{P_n})$, the agent maximizes aggregate expected payment minus cost. Principal $i$'s utility is $V_{P_i}(\bm{w}; F) = u_{P_i}(a^*) - \mathbb{E}_{x \sim F(a^*)}[w_{P_i}(x)]$, where $a^*$ is the agent's best response. As in \Cref{sec:robustness}, principals face distributional ambiguity, knowing only $\Ex[x \mid a] = a$. Each principal maximizes her worst-case payoff:
\[
V_{P_i}(\bm{w}) = \inf_{F \text{ compatible}} V_{P_i}(\bm{w}; F).
\]

\begin{restatable}{theorem}{thmCommonAgency}
\label{thm:robust-common-agency}
For any contract profile $(w_{P_1}, \ldots, w_{P_n})$, there exists an affine contract profile $({\psi}_{P_1}, \ldots, {\psi}_{P_n})$ such that $V_{P_i}(\bm{\psi}) \geq V_{P_i}(\bm{w})$ for all $i \in [n]$.
\end{restatable}

The proof applies concavification to the aggregate contract $w = \sum_{i=1}^n w_{P_i}$, finds a self-inducing action, then decomposes the supporting affine contract across principals using Carathéodory's theorem. However, the theorem is silent as to whether a robust equilibrium indeed exists when all principals use linear contracts. In fact, it is well-known that equilibria are hard to characterize due to non-convexity of principal's best responses (\cite{fraysse1993common, bernheim1986common, carmona2009existence}). \cite{marku2024robust} assumed that principal $i$'s utility is the identity at task $i$, and used a potential method to show equilibrium existence. In \Cref{sec:common-agency-characterization}, we characterize equilibria under homogeneous utility and cost functions, hence providing an alternative approach towards showing equilibrium existence.

\subsection{Multiple Agents: Team Production}
\label{sec:team-production}

In the team production setting, a single principal contracts with $n$ agents. Each agent $i$ exerts scalar effort $a_{A_i} \in [0, a_{\max}]$, and the team produces stochastic output $x \in \mathcal{X} \subseteq \mathbb{R}_+$ with $\Ex[x \mid \bm{a}] = f(\bm{a})$. The principal offers contracts $w_{A_i}: \mathcal{X} \to \mathbb{R}_+$ specifying payments as functions of realized output.

In this setting, we consider ambiguity-averse \emph{agents} who evaluate contracts under worst-case compatible distributions.

\begin{definition}[Ex-post Robust Equilibrium]
A strategy profile $\bm{\sigma}$ is an \emph{ex-post robust equilibrium} if it constitutes a Nash equilibrium for \emph{every} compatible distribution mapping $F$ satisfying $\Ex_{x \sim F(\bm{a})}[x] = f(\bm{a})$.
\end{definition}

We impose two assumptions: (i) each agent is essential ($a_{A_i} = 0$ implies $f(\bm{a}) = 0$), and (ii) there is distributional uncertainty except at zero effort ($f(a_{\max}, \ldots, a_{\max}) < x_{\max}$).

\begin{restatable}{theorem}{ThmTeamRobust}
\label{thm:team-robust}
Suppose the principal offers contracts $(w_{A_1}, \ldots, w_{A_n})$ that are not all affine. Under the above assumptions, the only ex-post robust equilibrium is $\bm{0}$.
\end{restatable}

The intuition is that non-affine contracts create a gap between the concave and convex envelopes of payments. Ambiguity-averse agents, when faced with such uncertainty, are driven toward zero effort—the only point where this uncertainty vanishes.

\section{Characterization of Optimal/Equilibrium Solutions with Homogeneous Functions}
\label{sec:homogeneous}

In general, optimal linear contracts are difficult to characterize in closed form. However, when the principal's utility and the agent's cost functions exhibit homogeneity, we obtain tractable characterizations across all three settings. We also derive computation and approximation guarantee results based on the characterizations. Proofs for this section are deferred to \Cref{proof:homogeneous-results}. 


\subsection{Optimal Contracts in Bilateral Setting}

We first consider the bilateral multitasking setting as in \Cref{sec:problem-statement}, and show that under homogeneity, the principal's optimization problem reduces to a tractable concave program. The key observation leverages Euler's theorem for homogeneous functions.

\begin{observation}\label{obs:euler-payment}
If $c$ is homogeneous of degree $k_c$, then for any action $a$ induced by a linear contract $\phi = \nabla c(a)$, the expected payment to the agent satisfies
\[
\ip{\phi}{a} = \ip{\nabla c(a)}{a} = k_c \cdot c(a).
\]
\end{observation}

This follows immediately from Euler's theorem. Consequently, the principal's utility when inducing action $a$ via a linear contract is
\[
u(a) - \ip{\phi}{a} = u(a) - k_c \cdot c(a).
\]

\begin{theorem}[Optimal Linear Contract in Bilateral Setting]\label{thm:bilateral-optimal}
Suppose $u$ is concave and $c$ is homogeneous of degree $k_c > 1$. The optimal action $a^*$ to induce via a linear contract solves
\[
a^* = \arg\max_{a} \left\{ u(a) - k_c \cdot c(a) \right\}.
\]
The optimal linear contract is $\phi^* = \nabla c(a^*)$, and this optimization is a concave program.
\end{theorem}


\paragraph{Approximation guarantees.} We now quantify the efficiency loss from using linear contracts relative to first-best. Define:
\begin{align*}
V^*_{\LIN} &= \max_\phi \left\{ u(a(\phi)) - \ip{\phi}{a(\phi)} \right\} = \max_a \left\{ u(a) - k_c \cdot c(a) \right\}, \\
V^*_{\FB} &= \max_a \left\{ u(a) - c(a) \right\}.
\end{align*}

\begin{restatable}{theorem}{ThmBilateralHomo}
\label{thm:approximation-bilateral}
Suppose the principal's utility $u(a)$ is homogeneous of degree $k_u \in (0,1)$ and concave, and the agent's cost $c(a)$ is homogeneous of degree $k_c > 1$ and convex. Then
\[
\frac{V^*_{\LIN}}{V^*_{\FB}} = k_c^{-k_u/(k_c-k_u)}.
\]
\end{restatable}

The proof analyzes the first-order conditions in both settings. At the first-best optimum $a^*_{\FB}$:
\[
\nabla u(a^*_{\FB}) = \nabla c(a^*_{\FB}).
\]
At the second-best optimum $a^*_{\LIN}$ under linear contracts:
\[
\nabla u(a^*_{\LIN}) = k_c \cdot \nabla c(a^*_{\LIN}).
\]
By homogeneity of the gradients, the second-best action is a scaled version of the first-best action: $a^*_{\LIN} = \lambda \cdot a^*_{\FB}$ for some $\lambda \in (0,1)$. Solving for $\lambda$ and applying Euler's theorem yields the exact approximation ratio. 


\subsection{Equilibrium in Common Agency}
\label{sec:common-agency-characterization}
We next study the common agency setting as in \Cref{sec:common-agency}. We characterize the equilibria in this setting under homogeneity.

\begin{assumption}\label{assumption:common-agency}
The cost function $c$ is homogeneous of degree $k_c > 1$. The utility functions $u_{P_1}, \ldots, u_{P_n}$ are concave and continuously differentiable.
\end{assumption}

\begin{theorem}[Equilibrium in Linear Contracts]\label{thm:common-agency-equilibrium}
Under \Cref{assumption:common-agency}, there exists a linear contract equilibrium $(\phi_{P_1}^*, \ldots, \phi_{P_n}^*)$ that is a robust Nash equilibrium. The induced action ${a}^*$ satisfies
\begin{equation}\label{eq:equilibrium-action}
(n(k_c - 1) + 1) \nabla c({a}^*) = \sum_{i=1}^n \nabla u_{P_i}({a}^*),
\end{equation}
and the equilibrium contracts are
\begin{equation}\label{eq:phi-star}
{\phi}_{P_i}^* = (1 - k_c) \nabla c({a}^*) + \nabla u_{P_i}({a}^*) \quad \text{for all } i \in [n].
\end{equation}
\end{theorem}

The equilibrium exhibits a natural structure: each principal's contract reflects her own marginal benefit, adjusted by a common term proportional to marginal cost. As the number of principals grows, the coefficient $n(k_c - 1) + 1$ increases, reflecting greater free-riding and reduced equilibrium effort.

\begin{proposition}[Tractable Computation]\label{prop:computation}
The equilibrium action ${a}^*$ from \Cref{eq:equilibrium-action} solves the concave maximization problem
\[
{a}^* = \arg\max_{{a}} \left\{ \sum_{i=1}^n u_{P_i}({a}) - (n(k_c - 1) + 1) c({a}) \right\}.
\]
Then given ${a}^*$, the equilibrium contracts follow from \eqref{eq:phi-star}.
\end{proposition}

Since $\sum_{i=1}^n u_{P_i}$ is concave and $c$ is convex, the objective is concave, making the equilibrium efficiently computable. Further, based on the equilibrium characterization, we give approximation guarantees for linear contracts, which appear in \Cref{sec:common-agency-app}.

\subsection{Optimal Linear Contracts in Team Production}

Finally, we study the team production setting from \Cref{sec:team-production}. We first show the game admits a potential function. 

\begin{proposition}[Potential Game Structure]\label{prop:potential-game}
Given a linear contract profile $\bm{\phi} = (\phi_{A_1}, \ldots, \phi_{A_n})$ with $\phi_{A_i} > 0$, the induced game among agents is a weighted potential game with potential function
\[
\Gamma(\bm{a}; \bm{\phi}) = f(\bm{a}) - \sum_{i=1}^n \frac{c_{A_i}(a_{A_i})}{\phi_{A_i}}.
\]
\end{proposition}

\begin{proof}
For each agent $i$ and any unilateral deviation from $a_{A_i}$ to $\bar{a}_{A_i}$:
\[
V_{A_i}(\bar{a}_{A_i}, {a}_{-i}; \bm{\phi}) - V_{A_i}(a_{A_i}, {a}_{-i}; \bm{\phi}) = \phi_{A_i} \left[ \Gamma(\bar{a}_{A_i}, {a}_{-i}; \bm{\phi}) - \Gamma(a_{A_i}, {a}_{-i}; \bm{\phi}) \right].
\]
Thus utility differences are proportional to potential differences, confirming the weighted potential game structure.
\end{proof}

When the production function $f$ is strictly concave, the potential $\Gamma$ is strictly concave in $\bm{a}$, so the unique equilibrium is
\[
\bm{a}(\bm{\phi}) = \arg\max_{\bm{a}} \Gamma(\bm{a}; \bm{\phi}).
\]

We focus on the case where the production function $f(\bm{a})$ is homogeneous of degree $k \in (0,1)$ and strictly concave. For simplicity, we normalize agent costs to $c_{A_i}(a_{A_i}) = a_{A_i}$; this is without loss of generality via an appropriate change of variables.

Define $\Gamma^*$ as the maximum potential achievable under budget-balanced linear contracts:
\[
\Gamma^* = \max_{\bm{\phi} \in \Delta} \Gamma(\bm{a}(\bm{\phi}); \bm{\phi}),
\]
where $\Delta = \{\bm{\phi} \geq 0 : \sum_i \phi_{A_i} = 1\}$ denotes budget balance.

\begin{theorem}\label{prop:team-production-optimal-gross-output}
Under the homogeneity assumptions above:
\begin{enumerate}
    \item The equilibrium potential satisfies $\Gamma(\bm{a}(\bm{\phi}); \bm{\phi}) = (1-k) f(\bm{a}(\bm{\phi}))$.
    \item The optimal potential admits the closed-form characterization
    \[
    \Gamma^* = \max_{\bm{a} \geq 0} \left\{ f(\bm{a}) - \left(\sum_i \sqrt{a_{A_i}}\right)^2 \right\}.
    \]
\end{enumerate}
\end{theorem}

The first part follows from Euler's theorem applied to the equilibrium conditions. The second part derives from optimizing over the contract shares $\bm{\phi} \in \Delta$: for fixed effort profile $\bm{a}$, the optimal shares minimize $\sum_i a_{A_i}/\phi_{A_i}$ subject to $\sum_i \phi_{A_i} = 1$, yielding $\phi_{A_i}^* \propto \sqrt{a_{A_i}}$ with minimum value $(\sum_i \sqrt{a_{A_i}})^2$. Based on this characterization, we give approximation guarantees and computational results on linear contracts, which appear in \Cref{sec:team-production-app}.

\section{Conclusion}
This paper establishes that linear contracts are worst-case optimal when the principal knows only the expected outcome $E[x|a]$ for each action, but not the full distribution. Our proof combines a concavification approach with the characterization of self-inducing actions. We show that a linear contract is worst-case optimal among all distributions consistent with the known first-moment structure. This result complements \cite{carroll2015robustness} analysis of technology set ambiguity with a distributional ambiguity perspective. We extend the concavification approach to multi-principal and multi-agent environments: in common agency, affine contracts improve all principals' worst-case payoffs, while in team production, non-affine contracts force zero effort. Lastly, we showed that imposing homogeneity enables tractable characterizations, and give approximation guarantees as well as computation results across all settings.

\newpage
\bibliographystyle{apalike}  
\bibliography{references}  

\appendix

\section{Proofs of Main Robustness Result in Bilateral Setting}

\begin{proof}[Proof of \Cref{lem:existence}]
Consider the optimization problem
\[
\min_{a \in \A} \left[ c(a) - \wbar(a) \right].
\]
First note that a minimizer exists. The function $c(a)$ is strictly convex and $\wbar(a)$ is concave (as the concave envelope of $w$). Hence the objective $c(a) - \wbar(a)$ is strictly convex. Since the domain $\A = \operatorname{conv}(\X)$ is compact, a minimizer exists.
Denote $\subdiff$ the subdifferential (for convex functions), $\supdiff$ the superdifferential (for concave functions). At the minimizer $a^\sharp$, we have $0 \in \subdiff[c(a) - \wbar(a)]$. Hence $\subdiff c(a) \cap \supdiff \wbar(a) \neq \varnothing$. This precisely means that $a^\sharp$ is self-inducing.
\end{proof}
\begin{proof}[Proof of \Cref{lem:distribution}]
Let $\psi^\sharp$ support $\wbar(a^\sharp)$ from above. Define the contact set
\[
P = \{x \in \X : \psi^\sharp(x) = w(x)\},
\]
which consists of outcomes where the affine contract $\psi^\sharp$ touches the original contract $w$. Since $\psi^\sharp$ is a supergradient of $\wbar$ at $a^\sharp$, we have $a^\sharp \in \operatorname{conv}(P)$.
Write $a^\sharp = \sum_{x \in P} \lambda_x x$ for some $\lambda_x \geq 0$ with $\sum_{x \in P} \lambda_x = 1$. Define the distribution $F$ by setting $F(a^\sharp)$ to be supported on $P$ with $\Pr(x) = \lambda_x$. Then $\E_{F(a^\sharp)}[x] = a^\sharp$, so $F$ is compatible.
For any action $a \in \A$, the agent's payoff under $w$ is at most $\wbar(a) - c(a)$ (since $\wbar(a)$ is the maximum expected payment over all compatible distributions for action $a$). Since $\phi \in \supdiff \wbar(a^\sharp)$:
\[
\wbar(a) \leq \wbar(a^\sharp) + \inner{\phi}{a - a^\sharp}.
\]
Since $\phi \in \subdiff c(a^\sharp)$:
\[
c(a) \geq c(a^\sharp) + \inner{\phi}{a - a^\sharp}.
\]
Combining:
\[
\wbar(a) - c(a) \leq \wbar(a^\sharp) - c(a^\sharp),
\]
and the inequality is strict for any $a \neq a^\sharp$.
At action $a^\sharp$ under distribution $F(a^\sharp)$, the agent's payoff is exactly $\E_{F(a^\sharp)}[w(x)] - c(a^\sharp)$. Since $F(a^\sharp)$ is supported on $P$ where $w(x) = \psi(x)$:
\[
\E_{F(a^\sharp)}[w(x)] = \E_{F(a^\sharp)}[\psi(x)] = \wbar(a^\sharp).
\]
Thus the agent's payoff at $a^\sharp$ equals $\wbar(a^\sharp) - c(a^\sharp)$, which is maximal. Hence $w$ induces $a^\sharp$ under $F$.
\end{proof}

\Cref{lem:affine-to-linear}] is adapted from \cite{roth2016watch}. We defer this to \Cref{proof:affine-to-linear}.

\begin{proof}[Proof of \Cref{thm:robustness}]
Let $w : \X \to \R_+$ be any contract.
By \Cref{lem:existence}, there exists a self-inducing action $a^\sharp$ with corresponding affine contract $\psi^\sharp$ such that $\psi$ induces $a^\sharp$ and $\psi^\sharp(x) \geq w(x)$ for all $x \in \X$. By \Cref{lem:distribution}, there exists a compatible distribution $F$ such that:
\begin{itemize}
    \item Both $w$ and $\psi^\sharp$ induce action $a^\sharp$ under $F$.
    \item The principal's payoff satisfies $V_P(\psi^\sharp; F) \geq V_P(w; F)$.
\end{itemize}
Under any compatible distribution $F'$, the agent's expected payment from affine contract $\psi^\sharp$ when taking action $a$ is
\[
\E_{F'(a)}[\psi(x)] = \psi_0^\sharp + \inner{\psi_{1:d}^\sharp}{\E_{F'(a)}[x]} = \psi_0^\sharp + \inner{\psi_{1:d}^\sharp}{a}.
\]
This depends only on $a$, not on the specific distribution $F'$. Therefore, the principal's worst-case payoff under $\psi^\sharp$ equals her payoff under any particular compatible distribution:
\[
V_P(\psi^\sharp) = V_P(\psi^\sharp; F).
\]
Then we have
\[
V_P(\psi^\sharp) = V_P(\psi^\sharp; F) \geq V_P(w; F) \geq \min_{F' \text{ compatible}} V_P(w; F') = V_P(w).
\]
By \Cref{lem:affine-to-linear}, among all affine contracts inducing action $a^\sharp$, the linear contract $\phi = \nabla c(a^\sharp)$ achieves the minimum expected payment. Since both $\psi^\sharp$ and $\phi$ induce $a^\sharp$, and $\phi$ pays weakly less:
\[
V_P(\phi) = u(a^\sharp) - \inner{\nabla c(a^\sharp)}{a^\sharp} \geq u(a^\sharp) - \left( \inner{\psi_{1:d}^\sharp}{a^\sharp} + \psi_0^\sharp \right) = V_P(\psi^\sharp).
\]
Combining all steps:
\[
V_P(\phi) \geq V_P(\psi^\sharp) \geq V_P(w).
\]
Since $w$ was arbitrary, linear contracts are optimal for the principal's worst-case objective.
\end{proof}

\section{Proof of Multilateral Robustness}
\begin{proof} [{Proof of \Cref{thm:robust-common-agency}}]

Let $(w_1, w_2)$ be any pair of contracts. Define the aggregate contract $w = w_1 + w_2$. Since the agent's expected utility under $(w_1, w_2)$ depends only on the aggregate payment, the agent's behavior is determined by $w$.
By the analysis in \Cref{sec:robustness}, applied to the aggregate contract $w$, there exists a self-inducing action $a^\sharp$ with corresponding affine contract $\psi$ such that:
\begin{itemize}
    \item $\psi$ supports the concave envelope $\overline{w}$ at $a^\sharp$ from above, and
    \item $\psi$ induces action $a^\sharp$.
\end{itemize}
By Carathéodory's theorem, there exists a compatible distribution $F$ supported on at most $d+1$ points $\{x_1, \ldots, x_m\}$ with $m \leq d+1$ such that:
\begin{itemize}
    \item $\sum_{k=1}^m \lambda_k x_k = a^\sharp$ for some $\lambda_k \geq 0$ with $\sum_k \lambda_k = 1$, and
    \item $\psi(x_k) = w(x_k) = w_1(x_k) + w_2(x_k)$ for all support points $x_k$.
\end{itemize}
Since $m \leq d+1$ and an affine function on $\mathbb{R}^d$ has $d+1$ degrees of freedom, we can construct affine contracts $\psi_1, \psi_2$ such that:
\[
\psi_1(x_k) = w_1(x_k) \quad \text{and} \quad \psi_2(x_k) = w_2(x_k) \quad \text{for all } k = 1, \ldots, m.
\]
Under distribution $F$:
\begin{enumerate}
    \item The aggregate $\psi_1 + \psi_2$ agrees with $w_1 + w_2$ on all support points, so the agent's expected utility from any action is the same under $(\psi_1, \psi_2)$ as under $(w_1, w_2)$. Hence the agent takes action $a^\sharp$ under both contract pairs.
    
    \item Principal $i$'s expected payment under $\psi_i$ equals that under $w_i$:
    \[
    \Ex_{x \sim F(a^\sharp)}[\psi_i(x)] = \sum_{k=1}^m \lambda_k \psi_i(x_k) = \sum_{k=1}^m \lambda_k w_i(x_k) = \Ex_{x \sim F(a^\sharp)}[w_i(x)].
    \]
\end{enumerate}
Therefore, under distribution $F$, each principal's payoff is the same under $(\psi_1, \psi_2)$ as under $(w_1, w_2)$.
Now we compare worst-case payoffs. For affine contracts, the agent's expected payment depends only on the expected signal, so payoffs are invariant across all compatible distributions:
\[
V_{P_i}(\psi_1, \psi_2; F) = V_{P_i}(\psi_1, \psi_2; F') \quad \text{for all compatible } F, F'.
\]
For the original contracts $(w_1, w_2)$, the worst-case payoff satisfies:
\[
V_{P_i}(w_1, w_2) = \inf_{F' \text{ compatible}} V_{P_i}(w_1, w_2; F') \leq V_{P_i}(w_1, w_2; F) = V_{P_i}(\psi_1, \psi_2; F) = V_{P_i}(\psi_1, \psi_2).
\]
Hence both principals are weakly better off under the affine contract pair $(\psi_1, \psi_2)$. 
\end{proof}

\begin{proof}[Proof of \Cref{thm:team-robust}]
We proceed in two steps.

\medskip
\noindent\textit{Step 1: Any ex-post robust equilibrium is pure.}
Suppose $\bm{\sigma}$ is an ex-post robust equilibrium. Fix agent $i$ and suppose the other agents play $\bm{\sigma}_{-i}$. Agent $i$'s expected utility from action $a_{A_i}$ is
\[
\Ex_{\substack{x \sim F(a_{A_i}, \bm{a}_{-i}) \\ \bm{a}_{-i} \sim \bm{\sigma}_{-i}}}[w_{A_i}(x)] - c_{A_i}(a_{A_i}).
\]
Choose $F$ to be the distribution that realizes the upper concave envelope $\overline{w}_{A_i}$ of $w_{A_i}$. Then agent $i$'s payoff becomes
\[
\Ex_{\bm{a}_{-i} \sim \bm{\sigma}_{-i}}\bigl[\overline{w}_{A_i}(f(a_{A_i}, \bm{a}_{-i}))\bigr] - c_{A_i}(a_{A_i}).
\]
Since $\overline{w}_{A_i}$ is concave and $c_{A_i}$ is strictly convex, this expression is strictly concave in $a_{A_i}$ for any fixed $\bm{a}_{-i}$. Hence the optimal action is unique, implying that any ex-post robust equilibrium must be in pure strategies.

\medskip
\noindent\textit{Step 2: The only pure-strategy ex-post robust equilibrium is $\bm{0}$.}
Consider any non-zero pure strategy profile $\bm{a}$. We construct a compatible distribution under which $\bm{a}$ is not an equilibrium. Since at least one contract $w_{A_i}$ is not affine, its concave envelope $\overline{w}_{A_i}$ and convex envelope $\underline{w}_{A_i}$ differ, and that in particular $\overline{w}_{A_i}(f(\ba)) > \underline{w}_{A_i}(f(\ba))$. Any payment function $\hat{w}_{A_i}$ satisfying $\underline{w}_{A_i} \leq \hat{w}_{A_i} \leq \overline{w}_{A_i}$ can be realized by some compatible distribution.

Choose a distribution such that the realized payment function $\hat{w}_{A_i}$ is constant in a neighborhood of $f(\bm{a})$. Specifically, for some $\eps > 0$:
\[
\hat{w}_{A_i}\bigl(f(a_{A_i} - \eps, \bm{a}_{-i})\bigr) = \hat{w}_{A_i}\bigl(f(\bm{a})\bigr).
\]
Under this distribution, agent $i$ receives the same expected payment from actions $a_{A_i}$ and $a_{A_i} - \eps$, but incurs strictly lower cost at $a_{A_i} - \eps$. Hence agent $i$ strictly prefers to deviate, so $\bm{a}$ is not an equilibrium under this distribution.
\end{proof}

\section{Proofs for \Cref{sec:homogeneous}}
\label{proof:homogeneous-results}

\subsection{Bilateral Multitasking}
\label{proof:homogeneous-results-bilateral}
\begin{proof}[Proof of \Cref{thm:approximation-bilateral}]

In the first-best, the principal pays the agent $w(a) = c(a)$ when the agent chooses action $a$, and the optimal action $a^*$ is
\[
a^* \in \arg\max_{a} \left[ u(a) - c(a) \right].
\]
The first-order condition gives $\nabla u(a^*) = \nabla c(a^*)$.

By Euler's theorem for homogeneous functions, we have:
\begin{align*}
\langle \nabla u(a^*), a^* \rangle &= k_u u(a^*), \\
\langle \nabla c(a^*), a^* \rangle &= k_c c(a^*).
\end{align*}

Since $\nabla u(a^*) = \nabla c(a^*)$, we obtain $k_u u(a^*) = k_c c(a^*)$, which gives us
\[
c(a^*) = \frac{k_u}{k_c} u(a^*).
\]

The first-best optimal utility is:
\[
\mathtt{OPT}_{\text{FB}} = u(a^*) - c(a^*) = u(a^*) - \frac{k_u}{k_c} u(a^*) = u(a^*) \frac{k_c - k_u}{k_c}.
\]

In the \emph{second-best} scenario with \emph{linear contracts}, the agent maximizes $\langle \phi, a \rangle - c(a)$, yielding the first-order condition $\phi = \nabla c(a)$. By Euler's theorem, the payment to the agent is 
\[
\langle \phi, a \rangle = \langle \nabla c(a), a \rangle = k_c c(a).
\]

If the principal induces action $a$ from the agent, his surplus under a linear contract is then:
\[
u(a) - k_c c(a).
\]

The optimal linear contract induces action:
\[
\tilde{a} = \arg\max_{a} \left[ u(a) - k_c c(a) \right],
\]
with first-order condition $\nabla u(\tilde{a}) = k_c \nabla c(\tilde{a})$.

By homogeneity, $\nabla u(\lambda a) = \lambda^{k_u-1} \nabla u(a)$ and $\nabla c(\lambda a) = \lambda^{k_c-1} \nabla c(a)$. Since $\nabla u(a^*) = \nabla c(a^*)$, we know that $\tilde{a} = \lambda a^*$ for some $\lambda > 0$. Then:
\begin{align*}
\nabla u(\tilde{a}) &= \lambda^{k_u-1} \nabla u(a^*) = \lambda^{k_u-1} \nabla c(a^*), \\
k_c \nabla c(\tilde{a}) &= k_c \lambda^{k_c-1} \nabla c(a^*).
\end{align*}

For the first-order condition to hold:
\[
\lambda^{k_u-1} = k_c \lambda^{k_c-1} \implies \lambda^{k_u-k_c} = k_c \implies \lambda = k_c^{1/(k_u-k_c)}.
\]

Since $k_u < k_c$, we have $k_u - k_c < 0$, so $\lambda = k_c^{-1/(k_c-k_u)} < 1$.

The linear contract optimal utility is:
\begin{align*}
\mathtt{LIN}_{\text{FB}} &= u(\tilde{a}) - k_c c(\tilde{a}) \\
&= u(\lambda a^*) - k_c c(\lambda a^*) \\
&= \lambda^{k_u} u(a^*) - k_c \lambda^{k_c} c(a^*) \\
&= \lambda^{k_u} u(a^*) - k_c \lambda^{k_c} \cdot \frac{k_u}{k_c} u(a^*) \\
&= u(a^*) (\lambda^{k_u} - k_u \lambda^{k_c}).
\end{align*}

The approximation ratio is:
\begin{align*}
\frac{\mathtt{LIN}_{\text{FB}}}{\mathtt{OPT}_{\text{FB}}} &= \frac{u(a^*) (\lambda^{k_u} - k_u \lambda^{k_c})}{u(a^*) \frac{k_c - k_u}{k_c}} \\
&= \frac{k_c(\lambda^{k_u} - k_u \lambda^{k_c})}{k_c - k_u}.
\end{align*}

Substituting $\lambda = k_c^{-1/(k_c-k_u)}$, we have:
\begin{align*}
\lambda^{k_u} &= k_c^{-k_u/(k_c-k_u)}, \\
\lambda^{k_c} &= k_c^{-k_c/(k_c-k_u)}.
\end{align*}

Note that 
\[
-\frac{k_u}{k_c-k_u} = -\frac{k_c}{k_c-k_u} + \frac{k_c-k_u}{k_c-k_u} = -\frac{k_c}{k_c-k_u} + 1,
\]
which implies $\lambda^{k_u} = k_c \cdot \lambda^{k_c}$. Therefore:
\begin{align*}
\lambda^{k_u} - k_u \lambda^{k_c} &= k_c \lambda^{k_c} - k_u \lambda^{k_c} \\
&= (k_c - k_u) \lambda^{k_c} \\
&= (k_c - k_u) k_c^{-k_c/(k_c-k_u)}.
\end{align*}

Substituting back:
\begin{align*}
\frac{\mathtt{LIN}_{\text{FB}}}{\mathtt{OPT}_{\text{FB}}} &= \frac{k_c \cdot (k_c - k_u) k_c^{-k_c/(k_c-k_u)}}{k_c - k_u} \\
&= k_c \cdot k_c^{-k_c/(k_c-k_u)} \\
&= k_c^{1 - k_c/(k_c-k_u)} \\
&= k_c^{(k_c-k_u-k_c)/(k_c-k_u)} \\
&= k_c^{-k_u/(k_c-k_u)}. 
\end{align*} 
This finishes the proof. 
\end{proof}

\subsection{Common Agency}
\label{sec:common-agency-app}

\begin{proof} [Proof of \Cref{thm:common-agency-equilibrium}]
We characterize the equilibrium by deriving first-order conditions for all parties and showing that a solution exists. For simplicity we consider the two principal setting, though the idea extends easily to the more general case. 

\paragraph{Agent's optimization.}
Given linear contracts $\phi_1, \phi_2 \in (\mathbb{R}_+)^d$ from the two principals, the agent's expected utility is
\[
V_A(a; \phi_1, \phi_2) = \langle \phi_1 + \phi_2, a \rangle - c(a).
\]
The agent's first-order condition is:
\begin{equation}
\label{eq:agent-foc}
\phi_1 + \phi_2 = \nabla c(a).
\end{equation}

\paragraph{Principal $i$'s optimization.}
Fix principal $j$'s linear contract $\phi_j$. Principal $i$ chooses which action $a$ to induce. By \eqref{eq:agent-foc}, inducing action $a$ requires $\phi_i = \nabla c(a) - \phi_j$. Principal $i$'s payoff is:
\[
V_{P_i}(a \mid \phi_j) = u_i(a) - \langle \nabla c(a) - \phi_j, a \rangle = u_i(a) - \langle \nabla c(a), a \rangle + \langle \phi_j, a \rangle.
\]
Since $c$ is homogeneous of degree $k$, Euler's theorem gives $\langle \nabla c(a), a \rangle = k \cdot c(a)$. Substituting:
\[
V_{P_i}(a \mid \phi_j) = u_i(a) - k \cdot c(a) + \langle \phi_j, a \rangle.
\]
Taking the gradient with respect to $a$ and setting it to zero yields principal $i$'s first-order condition:
\begin{equation}
\label{eq:principal-foc}
\phi_j = k \cdot \nabla c(a) - \nabla u_i(a).
\end{equation}
By symmetry, principal $j$'s first-order condition is:
\begin{equation}
\label{eq:principal-foc-j}
\phi_i = k \cdot \nabla c(a) - \nabla u_j(a).
\end{equation}

\paragraph{Equilibrium characterization.}
Adding equations \eqref{eq:agent-foc}, \eqref{eq:principal-foc}, and \eqref{eq:principal-foc-j}, we obtain:
\[
\nabla c(a) = \frac{\nabla u_1(a) + \nabla u_2(a)}{2k - 1}.
\]
The above equation holds at $a^*$, where $a^*$ is the minimizer to 
\[
(2k-1) c(a) - u_{P_1}(a) - u_{P_2}(a). 
\]


Given $a^*$, we define the equilibrium contracts:
\[
\phi_1^* = k \cdot \nabla c(a^*) - \nabla u_2(a^*), \qquad \phi_2^* = k \cdot \nabla c(a^*) - \nabla u_1(a^*).
\]
Then $\phi_{P_1}^*$ and $\phi_{P_2}$ is the unique Nash equilibrium. 
\end{proof}

Given the characterization of the unique equilibrium, we can now analyze the approximation ratio. Here, we compare the optimal \emph{social welfare} in the first-best scenario vs.\ under the equilibrium when using linear contracts. The first-best social welfare is defined as: 
\[
SW^*_{\FB} = \max_a u_{P_1}(a) + u_{P_2}(a) - c(a). 
\]
And the second-best using linear contracts is defined as:
\[
SW^*_{\LIN} = u_{P_1}(a) + u_{P_2}(a) - c(a),
\]
where $a$ is the action induced in the equilibrium as defined in \Cref{thm:common-agency-equilibrium}. 

We assume the cost function $c$ is homogeneous of degree $k_c$. The utility functions $u_{P_1}, u_{P_2}$ are homogeneous of degree $k_u$.

\begin{theorem}
\label{thm:approximation-common-agency}
Under the above assumption, the approximation ratio is 
\[
\frac{SW^*_{\LIN}}{SW^*_{\FB}} = \frac{k_c (2k_c - 1)^{-\frac{k_u}{k_c - k_u}} - k_u (2k_c - 1)^{-\frac{k_c}{k_c - k_u}}}{k_c - k_u}.
\]
\end{theorem}

Similarly as in the bilateral setting, the analysis shows that the agent's action in first-best is a scaled version of the action in the linear contract equilibrium (by homogeneity). Specifically, $a^*_{\mathsf{LIN}} = \gamma \cdot a^*_{\mathsf{FB}}$ where $\gamma = (2k_c - 1)^{-1/(k_c - k_u)} < 1$. Applying the first-order conditions together with Euler's theorem for homogeneous functions yields the approximation ratio. 

\begin{proof}[Proof of \Cref{thm:approximation-common-agency}]
We characterize the relation between the optimal action to induce under first-best and second-best, then apply Euler's theorem to obtain the approximation ratio. 
\paragraph{First-Best Analysis.}
In the first-best, the social planner solves
\[
a^*_{\mathrm{FB}} = \arg\max_a \; u_{P_1}(a) + u_{P_2}(a) - c(a).
\]
The first-order condition at $a^*_{\mathrm{FB}}$ is
\[
\nabla u_{P_1}(a^*_{\mathrm{FB}}) + \nabla u_{P_2}(a^*_{\mathrm{FB}}) = \nabla c(a^*_{\mathrm{FB}}).
\]
Multiplying both sides by $a^*_{\mathrm{FB}}$ and applying Euler's theorem for homogeneous functions (i.e., $a \cdot \nabla f(a) = k \cdot f(a)$ for $f$ homogeneous of degree $k$), we obtain
\[
k_u \bigl[u_{P_1}(a^*_{\mathrm{FB}}) + u_{P_2}(a^*_{\mathrm{FB}})\bigr] = k_c \, c(a^*_{\mathrm{FB}}).
\]
Rearranging:
\[
u_{P_1}(a^*_{\mathrm{FB}}) + u_{P_2}(a^*_{\mathrm{FB}}) = \frac{k_c}{k_u} \, c(a^*_{\mathrm{FB}}).
\]
Thus, the first-best social welfare is
\[
\mathrm{SW}_{\mathrm{FB}} = \frac{k_c}{k_u} \, c(a^*_{\mathrm{FB}}) - c(a^*_{\mathrm{FB}}) = \frac{k_c - k_u}{k_u} \, c(a^*_{\mathrm{FB}}).
\]

\paragraph{Linear Contract Equilibrium.}
Under linear contracts, the equilibrium action $a^*_{\mathrm{LIN}}$ satisfies
\[
\nabla c(a^*_{\mathrm{LIN}}) = \frac{\nabla u_{P_1}(a^*_{\mathrm{LIN}}) + \nabla u_{P_2}(a^*_{\mathrm{LIN}})}{2k_c - 1}.
\]

Since $\nabla c$ is homogeneous of degree $k_c - 1$ and $\nabla u_{P_i}$ is homogeneous of degree $k_u - 1$, we posit that $a^*_{\mathrm{LIN}} = \gamma \, a^*_{\mathrm{FB}}$ for some scalar $\gamma > 0$. Substituting into the equilibrium condition:
\[
\gamma^{k_c - 1} \nabla c(a^*_{\mathrm{FB}}) = \frac{\gamma^{k_u - 1}}{2k_c - 1} \bigl[\nabla u_{P_1}(a^*_{\mathrm{FB}}) + \nabla u_{P_2}(a^*_{\mathrm{FB}})\bigr].
\]
Using the first-best FOC, $\nabla u_{P_1}(a^*_{\mathrm{FB}}) + \nabla u_{P_2}(a^*_{\mathrm{FB}}) = \nabla c(a^*_{\mathrm{FB}})$:
\[
\gamma^{k_c - 1} = \frac{\gamma^{k_u - 1}}{2k_c - 1}.
\]
Solving for $\gamma$:
\[
\gamma^{k_c - k_u} = \frac{1}{2k_c - 1}, \quad \text{hence} \quad \gamma = (2k_c - 1)^{-\frac{1}{k_c - k_u}}.
\]

Consequently:
\[
a^*_{\mathrm{LIN}} = (2k_c - 1)^{-\frac{1}{k_c - k_u}} \, a^*_{\mathrm{FB}}.
\]

\textbf{Social Welfare under Linear Contracts.}
Using homogeneity:
\begin{align*}
u_{P_i}(a^*_{\mathrm{LIN}}) &= \gamma^{k_u} \, u_{P_i}(a^*_{\mathrm{FB}}), \\
c(a^*_{\mathrm{LIN}}) &= \gamma^{k_c} \, c(a^*_{\mathrm{FB}}).
\end{align*}
Therefore:
\begin{align*}
\mathrm{SW}_{\mathrm{LIN}} &= \gamma^{k_u} \bigl[u_{P_1}(a^*_{\mathrm{FB}}) + u_{P_2}(a^*_{\mathrm{FB}})\bigr] - \gamma^{k_c} \, c(a^*_{\mathrm{FB}}) \\
&= \gamma^{k_u} \cdot \frac{k_c}{k_u} \, c(a^*_{\mathrm{FB}}) - \gamma^{k_c} \, c(a^*_{\mathrm{FB}}) \\
&= c(a^*_{\mathrm{FB}}) \left( \frac{k_c}{k_u} \gamma^{k_u} - \gamma^{k_c} \right).
\end{align*}

\paragraph{Approximation Ratio.}
The approximation ratio is
\[
\rho = \frac{\mathrm{SW}_{\mathrm{LIN}}}{\mathrm{SW}_{\mathrm{FB}}} = \frac{\frac{k_c}{k_u} \gamma^{k_u} - \gamma^{k_c}}{\frac{k_c - k_u}{k_u}} = \frac{k_c \gamma^{k_u} - k_u \gamma^{k_c}}{k_c - k_u}.
\]
Substituting $\gamma = (2k_c - 1)^{-\frac{1}{k_c - k_u}}$:
\[
\boxed{\rho = \frac{k_c (2k_c - 1)^{-\frac{k_u}{k_c - k_u}} - k_u (2k_c - 1)^{-\frac{k_c}{k_c - k_u}}}{k_c - k_u}.}
\]
This finishes the proof. 
\end{proof}

\subsection{Team Production}
\label{sec:team-production-app}

\begin{proof}[Proof of \Cref{prop:team-production-optimal-gross-output}]
We begin with part 1. At the agent's best response $a(\phi)$, the first-order conditions give $\phi_i \frac{\partial f}{\partial a_{A_i}} = 1$, so $\frac{\partial f}{\partial a_{A_i}} = 1/\phi_i$. Together with Euler's theorem for homogeneous functions:
\begin{align*}
\Gamma(\phi) &= f(a(\phi)) - \sum_i a_{A_i} / \phi_i \\
&= f(a(\phi)) - \sum_i a_{A_i} \frac{\partial f}{\partial a_{A_i}} \\
&= f(a(\phi)) - k f(a(\phi)) \\
&= (1-k) f(a(\phi)). \
\end{align*}
For part 2, notice
\begin{align*}
\Gamma^* &:= \max_{\phi \in \Delta} \Gamma(\phi) \\
&= \max_{\phi \in \Delta} \max_a f(a) - \sum_i a_{A_i} / \phi_i \\
&= \max_a \max_{\phi \in \Delta}  f(a) - \sum_i a_{A_i} / \phi_i
\end{align*}

For fixed $a$, the inner minimization over $\phi \in \Delta$ of $\sum_i a_{A_i}/\phi_i$ is achieved at $\phi_i = \sqrt{a_{A_i}}/\sum_j \sqrt{a_{A_j}}$, yielding value $\left(\sum_i \sqrt{a_{A_i}}\right)^2$. Therefore:
\begin{align*}
\Gamma^* &= \max_{a} f(a) - \left(\sum_i \sqrt{a_{A_i}}\right)^2. 
\end{align*} 
This finishes the proof. 
\end{proof}

Applying the Cauchy--Schwarz inequality to the \Cref{prop:team-production-optimal-gross-output} yields the following corollary.

\begin{cor}\label{cor:team-cauchy}
$\displaystyle \Gamma^* \geq \max_{\bm{a}} \left\{ f(\bm{a}) - n \sum_i a_{A_i} \right\}$.
\end{cor}



\subsubsection{Computation}
We now consider the computation problem of finding the best set of linear contracts subject to budget balance. By \Cref{prop:team-production-optimal-gross-output}, this is equivalent to optimizing the potential function $\Gamma^*$, where we have already obtained a tractable characterization. 

We provide a partial characterization for the special class of supermodular and homogeneous production functions, delineating the boundary between computational tractability and intractability. We first establish a hardness result showing that no fully polynomial-time approximation scheme (FPTAS) exists.

\begin{theorem}[Computational Hardness]\label{thm:team-no-fptas}
Assume $f(\bm{a})$ is supermodular and homogeneous. The principal's optimization problem admits no FPTAS unless $\mathsf{P} = \mathsf{NP}$.
\end{theorem}

The prove is by a reduction to max independent set and a result by \citep{nesterov2003random}. 

\begin{proof} 
We reduce from the maximum independent set problem. Given a graph $G = (V, E)$ with $|V| = n$ vertices, we construct an instance with $n + \binom{n}{2} - |E|$ agents. The production function is defined as
\[
f(a) = \sum_{\{i,j\} \notin E} (a_{A_i})^{1/4} \cdot (a_{A_j})^{1/4} \cdot (a_{A_{ij}})^{1/4},
\]
where the agents $\set{i}\cup \set{{ij}}$ are indexed by vertices $i \in V$ and non-edges $\{i,j\} \notin E$.

We perform the change of variables $x_i := a_{A_i}^{1/4}$ for vertex agents and $y_{ij} := a_{A_{ij}}^{1/4}$ for non-edge agents. The principal's problem becomes
\[
\max_{x, y \geq 0} \left\{ \sum_{\{i,j\} \notin E} y_{ij} x_i x_j - \left(\sum_i x_i^2 + \sum_{\{i,j\} \notin E} y_{ij}^2\right)^2 \right\}.
\]

By Nesterov's theorem (\Cref{thm:nesterov}), the inner maximization over the unit sphere satisfies
\[
\max_{\|x\|^2 + \|y\|^2 = 1} \sum_{\{i,j\} \notin E} y_{ij} x_i x_j = \frac{\sqrt{1 - 1/\alpha(G)}}{3\sqrt{3}},
\]
where $\alpha(G)$ denotes the stability number (maximum independent set size) of $G$.

Substituting and optimizing over the scale $t := \|x\|^2 + \|y\|^2$, the principal's problem reduces to
\[
\max_{t \geq 0} \left\{ \frac{\sqrt{1 - 1/\alpha(G)}}{3\sqrt{3}} \cdot t^{3/2} - t^2 \right\}.
\]
Solving the first-order condition yields the optimal value
\[
\frac{\left(1 - 1/\alpha(G)\right)^2}{256 \cdot 27}.
\]
An FPTAS for the principal's problem would then enable polynomial-time computation of $\alpha(G)$. 
\end{proof}

\begin{theorem}[\citealp{nesterov2003random}]\label{thm:nesterov}
Let $G = (V, E)$ be a graph with stability number $\alpha(G)$. Then
\[
\sqrt{1 - \frac{1}{\alpha(G)}} = 3\sqrt{3} \cdot \max_{\|\bm{x}\|^2 + \|\bm{y}\|^2 = 1} \sum_{\{i,j\} \notin E} y_{ij} x_i x_j.
\]
\end{theorem}

Despite this hardness result, we identify a tractable regime: when the homogeneity degree is less than $1/2$, the principal's optimization problem reduces to concave maximization and becomes polynomially tractable.

\begin{theorem}[Tractability Regime]\label{thm:team-tractable-regime}
If $f(\bm{a})$ is supermodular and homogeneous of degree $k < 1/2$, then the principal's optimization problem can be solved via concave maximization.
\end{theorem}
\begin{proof} 
Apply the change of variables $z_i := \sqrt{a_{A_i}}$, so that $a_{A_i} = z_i^2$. The principal's objective becomes
\[
f(z_1^2, \ldots, z_n^2) - \left(\sum_i z_i\right)^2.
\]
Since $f$ is supermodular and homogeneous of degree $k$, the composition $g(\bm{z}) := f(z_1^2, \ldots, z_n^2)$ is homogeneous of degree $2k < 1$ in $\bm{z}$. By \Cref{thm:supermodular-concave}, any supermodular function with homogeneity degree $\ge 1$ is concave, hence $g(\bm{z})$ is concave. The second term $(\sum_i z_i)^2$ is convex, so the overall objective is concave in $\bm{z}$. The problem thus reduces to maximizing a concave function over the non-negative orthant, which is tractable.
\end{proof}



\subsubsection{Approximation Guarantee}
We now study the approximation guarantee of linear contracts. We first define the first-best gross output subject to budget balance: 
\[
V^*_{\FB} = \max_{\bm{a}} \left\{ f(\bm{a})  \right\}, \text{s.t. } f(\bm{a}) = \sum_i a_i. 
\]
The second-best problem using linear contracts is
\[
V^*_{\LIN} = \max_{\bm{\phi} \in \Delta} \left\{ f(\bm{a}(\bm{\phi})) \right\}.
\]
By \Cref{prop:team-production-optimal-gross-output},
\[
V^*_{\LIN} = \max_{\bm{a}\ge 0}\left[f(\bm{a}) - (\sum_i \sqrt{a_{A_i}})^2 \right] / (1-k)
\]
By \Cref{cor:team-cauchy}, 
\[
V^*_{\LIN} \ge \max_{\bm{a}\ge 0}\left[f(\bm{a}) - n \sum_i a_{A_i} \right] / (1-k)
\]

\begin{theorem}
The approximation ratio of linear contracts is
\[
\frac{V^*_{\LIN}}{V^*_{\FB}} \ge \left(\frac{k}{n}\right)^{\frac{k}{1-k}}.
\]
\end{theorem}

\begin{proof}
At the first best, by the Lagrangian condition and Euler's theorem:
\[
\nabla f(\bm{a}_{\FB}) = k \cdot \bm{1}.
\]

Now consider $\tilde{\bm{a}}$ where:
\[
\nabla f(\tilde{\bm{a}}) = n \cdot \bm{1}.
\]

Since $\nabla f$ is homogeneous of degree $k-1$, we have $\tilde{\bm{a}} = t \cdot \bm{a}_{\FB}$ where:
\[
t^{k-1} \cdot k = n \implies t = \left(\frac{k}{n}\right)^{\frac{1}{1-k}}.
\]

By Euler's theorem at $\tilde{\bm{a}}$:
\[
f(\tilde{\bm{a}}) - n\sum_i \tilde{a}_i = f(\tilde{\bm{a}}) - k \cdot f(\tilde{\bm{a}}) = (1-k)f(\tilde{\bm{a}}).
\]

Since $f(\tilde{\bm{a}}) = t^k V^*_{\FB}$:
\[
\frac{f(\tilde{\bm{a}}) - n\sum_i \tilde{a}_i}{1-k} = \left(\frac{k}{n}\right)^{\frac{k}{1-k}} V^*_{\FB}.
\]

Therefore:
\[
\frac{V^*_{\LIN}}{V^*_{\FB}} \ge \left(\frac{k}{n}\right)^{\frac{k}{1-k}}.
\]
\end{proof}

\section{Proof of Auxillary Results}

\subsection{Improving Affine to Linear Contracts}
\label{proof:affine-to-linear}
\begin{lemma}\label{lem:affine-to-linear}
Suppose the agent's action set $\A$ is downward-closed, meaning if $a \in \A$, then $\rho a \in \A$ for all $\rho \in (0,1)$. Among all non-negative affine contracts that induce a given action $a$, the linear contract $\phi = \nabla c(a)$ induces $a$ with the minimum expected payment to the agent.
\end{lemma}
\begin{proof} 
For any affine contract $\psi(x) = \psi_0 + \langle \psi_{1:d}, x \rangle$, the agent's expected utility from action $a$ is
\[
U(a; \psi) = \langle \psi_{1:d}, a \rangle + \psi_0 - c(a),
\]
which is strictly concave in $a$.
First, observe that the linear contract $\phi = \nabla c(a)$ indeed induces action $a$: the first-order condition $\nabla c(a) = \phi$ is satisfied with equality.
Now fix any action $a^\dagger$ and suppose affine contract $\psi^\dagger$ induces $a^\dagger$, i.e., $a^\dagger \in \arg\max_{a \in \mathcal{A}} U(a; \psi^\dagger)$. The first-order optimality condition requires
\[
\langle \psi^\dagger_{1:d} - \nabla c(a^\dagger), a - a^\dagger \rangle \leq 0, \quad \forall a \in \mathcal{A}.
\]
Since $\mathcal{A}$ is downward-closed, $a = a^\dagger / 2 \in \mathcal{A}$. Substituting:
\[
\langle \psi^\dagger_{1:d} - \nabla c(a^\dagger), -a^\dagger / 2 \rangle \leq 0,
\]
which implies
\[
\langle \psi^\dagger_{1:d}, a^\dagger \rangle \geq \langle \nabla c(a^\dagger), a^\dagger \rangle.
\]
The expected payment to the agent under $\psi^\dagger$ is $\langle \psi^\dagger_{1:d}, a^\dagger \rangle + \psi^\dagger_0$. Limited liability requires $\psi^\dagger_0 \geq 0$. Therefore,
\[
\langle \psi^\dagger_{1:d}, a^\dagger \rangle + \psi^\dagger_0 \geq \langle \nabla c(a^\dagger), a^\dagger \rangle + 0 = \langle \nabla c(a^\dagger), a^\dagger \rangle.
\]
The linear contract $\phi = \nabla c(a^\dagger)$ achieves this lower bound with payment exactly $\langle \nabla c(a^\dagger), a^\dagger \rangle$.
\end{proof}

\subsection{Choquet's Assertion}
\label{proof:choquet}
\begin{lemma}
Let $H(x)$ be a symmetric square matrix with diagonal entries non-positive and off-diagonal entries non-negative.
Suppose for each $x > 0$ there exists a positive vector $v$ such that $H(x)v = 0$.
Then $H(x)$ is negative semidefinite.
\end{lemma}

\begin{proof}
Let $D = \mathrm{diag}(v_1, \ldots, v_n)$ and define
\[
B = D^{-1} H(x) D = [b_{ij}], \qquad \text{so that } \ b_{ij} = \frac{h_{ij} v_j}{v_i}.
\]
Since $H(x)$ has non-positive diagonal entries and non-negative off-diagonal entries, it follows that
\[
b_{ii} \le 0, \quad b_{ij} \ge 0 \ \text{for all } i \neq j.
\]

From the condition $H(x)v = 0$, we have for each $i$,
\[
0 = (H(x)v)_i = \sum_{j} h_{ij} v_j = v_i \sum_{j} b_{ij},
\]
hence
\[
\sum_{j} b_{ij} = 0 \quad \Longrightarrow \quad b_{ii} = -\sum_{j \neq i} b_{ij}.
\]
Define the Gershgorin radius
\[
R_i = \sum_{j \neq i} |b_{ij}| = \sum_{j \neq i} b_{ij} = -b_{ii}.
\]

By the \emph{Gershgorin Circle Theorem}, every eigenvalue $\lambda$ of $B$ lies in one of the disks
\[
D_i = \{ z \in \mathbb{C} : |z - b_{ii}| \le R_i \}
     = \{ z : |z - b_{ii}| \le -b_{ii} \}.
\]
Each disk $D_i$ lies entirely in the closed left half-plane and touches the origin, since it corresponds to the interval $[2b_{ii}, 0]$ on the real line (with $b_{ii} < 0$).
Therefore, every eigenvalue of $B$ satisfies $\operatorname{Re}(\lambda) \le 0$.

Because $B$ and $H(x)$ are similar ($B = D^{-1} H(x) D$), they have the same eigenvalues.
Further since $H(x)$ is symmetric, all its eigenvalues are real, so $\lambda \le 0$ for all eigenvalues, implying that $H(x)$ is negative semidefinite.

Thus, under the stated conditions, $H(x)$ is negative semidefinite.
\end{proof}

\begin{theorem}\label{thm:supermodular-concave}
Let $f: \bbR_+^n \to \bbR_+$ be supermodular and homogeneous of degree $k \leq 1$. Then $f$ is concave.
\end{theorem}

\begin{proof}
We give the proof for the case when degree $=1$. The case when degree $<1$ simply follows from a function composition. 
Take the gradient of the Euler's formula. Let $H(x)$ be the Hessian matrix. Then we know that
\[
H(x) x = 0.
\]
Further since the function is supermodular, the matrix $H(x)$ has non-negative off-diagonal entries. We can then directly apply the above lemma and know that $f(x)$ is concave. 
\end{proof}

\end{document}